\documentclass{llncs}

\usepackage{amsmath, amssymb}
\usepackage{tikz}
\usetikzlibrary{arrows,automata}
\usepackage{subfig}
\usepackage{algorithmicx}
\usepackage{algpseudocode}

\usepackage{lineno}
\linenumbers
\setpagewiselinenumbers
 
\let\doendproof\endproof
\renewcommand\endproof{~\hfill$\qed$\doendproof}

\title{On the structure and syntactic complexity of generalized definite languages}
\author{Szabolcs Iv\'an \and Judit Nagy-Gy\"orgy}
\institute{University of Szeged, Hungary}

\begin{document}

\maketitle

\begin{abstract}
We give a forbidden pattern characterization for the class of generalized definite languages,
show that the corresponding problem is $\mathbf{NL}$-complete and can be solved in quadratic time.
We also show that their syntactic complexity coincides with that of the definite languages
and give an upper bound of $n!$ for this measure.
\end{abstract}

\section{Introduction}

A language is generalized definite if membership can be decided for a word by looking at its prefix and suffix of a given constant length.
Generalized definite languages and automata were introduced by Ginzburg~\cite{ginzburg} in 1966
and further studied in e.g.~\cite{ciricimrehsteinby,gecsegimreh,petkovic,steinby}.
This language class is strictly contained within the class of star-free languages, lying on the first level of the
dot-depth hierarchy~\cite{dotdepth}.
This class possess a characterization in terms of its syntactic semigroup \cite{perrin}: a regular language is generalized definite
if and only if its syntactic semigroup is locally trivial if and only if it satisfies a certain identity $x^\omega yx^\omega=x^\omega$.
This characterization is hardly efficient by itself when the language is given by its minimal automaton,
since the syntactic semigroup can be much larger than the automaton (a construction for a definite language with state complexity
-- that is, the number of states of its minimal automaton -- $n$
and syntactic complexity -- that is, the size of the transition semigroup of its minimal automaton --
$\lfloor e(n-1)!\rfloor$ is explicit in \cite{brzozo}).
However, as stated in~\cite{handbook}, Sec.~5.4, it is usually not necessary to compute the (ordered) syntactic semigroup
but most of the time one can develop a more efficient algorithm by analyzing the minimal automaton.
As an example for this line of research, recently, the authors of \cite{klima-polak} gave a nice characterization of minimal automata of
piecewise testable languages, yielding a quadratic-time decision algorithm, matching an alternative (but of course equivalent)
earlier (also quadratic) characterization of~\cite{trahtman} which improved the $\mathcal{O}(n^5)$ bound of~\cite{Stern}.

In this paper we give a forbidden pattern characterization for generalized definite languages in terms of the minimal automaton,
and analyze the complexity of the decision problem whether a given automaton recognizes a generalized definite language,
yielding an $\mathbf{NL}$-completeness result (with respect to logspace reductions) as well as a deterministic decision
procedure running in $\mathcal{O}(n^2)$ time (on a RAM machine).

There is an ongoing line of research for syntactic complexity of regular languages.
In general, a regular language with state complexity $n$ can have a syntactic complexity of $n^n$, already in the case
when there are only three input letters. There are at least two possible modifications of the problem:
one option is to consider the case when the input alphabet is binary (e.g. as done in \cite{holzer-konig,krawetz-lawrence-shallit}).
The second option is to study a strict subclass of regular languages.
In this case, the syntactic complexity of a class $\mathcal{C}$ of languages is a function $n\mapsto f(n)$,
with $f(n)$ being the maximal syntactic complexity a member of $\mathcal{C}$ can have whose state complexity
is (at most) $n$. The syntactic complexity of several language classes, e.g. (co)finite, reverse definite,
bifix--, factor-- and subword-free languages etc. is precisely determined in \cite{limsc}.
However, the exact syntactic complexity of the (generalized) definite languages and that of the star-free languages
(as well as the locally testable or the locally threshold testable languages) is not known yet.

We also address this problem and show that the syntactic complexity of generalized definite languages
coincides with that of definite languages, and show an upper bound $n!$ for this measure.
Since the lower bound is $\Omega((n-1)!)$, this is asymptotically optimal up to a logarithmic factor.

\section{Notation}

We assume the reader is familiar with the standard notions of automata and language theory,
but still we give a summary for the notation.

When $n\geq 0$ is an integer, $[n]$ stands for the set $\{1,\ldots,n\}$.
An \emph{alphabet} is a nonempty finite set $\Sigma$.
The set of \emph{words} over $\Sigma$ is denoted $\Sigma^*$, while $\Sigma^+$ stands for the set of \emph{nonempty words}.
The \emph{empty word} is denoted $\varepsilon$.
A \emph{language} over $\Sigma$ is an arbitrary set $L\subseteq\Sigma^*$ of $\Sigma$-words.

A (finite) \emph{automaton} (over $\Sigma$) is a system $\mathbb{A}=(Q,\Sigma,\delta,q_0,F)$ where
$Q$ is the finite set of states,
$q_0\in Q$ is the start state,
$F\subseteq Q$ is the set of final (or accepting) states,
and $\delta:Q\times \Sigma\to Q$ is the transition function.
The transition function $\delta$ extends in a unique way to a right action of the monoid $\Sigma^*$ on $Q$,
also denoted $\delta$ for ease of notation%
.
When $\delta$ is understood, we write $q\cdot u$, or simply $qu$ for $\delta(q,u)$.
Moreover, when $C\subseteq Q$ is a subset of states and $u\in\Sigma^*$ is a word, let $Cu$ stand for the set $\{pu:p\in C\}$
and when $L$ is a language, $CL=\{pu:p\in C,u\in L\}$.
The \emph{language recognized by $\mathbb{A}$} is $L(\mathbb{A})=\{x\in\Sigma^*:q_0x\in F\}$.
A language is \emph{regular} if it can be recognized by some finite automaton.

The state $q\in Q$ is \emph{reachable} from a state $p\in Q$ in $\mathbb{A}$, denoted $p\preceq_{\mathbb{A}} q$,
or just $p\preceq q$ if there is no danger of confusion, if $pu=q$ for some $u\in\Sigma^*$.
An automaton is \emph{connected} if its states are all reachable from its start state.

Two states $p$ and $q$ of $\mathbb{A}$ are \emph{distinguishable} if there exists a word $u\in\Sigma^*$ such that
exactly one of $pu$ and $qu$ belongs to $F$. In this case we say that $u$ \emph{separates} $p$ and $q$.
A connected automaton is called \emph{reduced} if each pair of distinct states is distinguishable.

It is known that for each regular language $L$ there exists a reduced automaton $\mathbb{A}_L$, unique up to isomorphism, recognizing $L$.
$\mathbb{A}_L$ can be computed from any automaton recognizing $L$ by an efficient algorithm called minimization and is called the
\emph{minimal automaton} of $L$.\marginpar{$\mathbb{A}_L$}

The classes of the equivalence relation $p\sim q\ \Leftrightarrow p\preceq q\textrm{ and }q\preceq p$ are called \emph{components} of $\mathbb{A}$.
A component $C$ is \emph{trivial} if $C=\{p\}$ for some state $p$ such that $pa\neq p$ for any $a\in\Sigma$,
and is a \emph{sink} if $C\Sigma\subseteq C$. It is clear that each automaton has at least one sink and sinks are never trivial.\marginpar{(trivial) components and sinks}
The \emph{component graph} $\Gamma(\mathbb{A})$ of $\mathbb{A}$ is an edge-labelled directed graph $(V,E,\ell)$ along with a mapping $c:Q\to V$
where $V$ is the set of the $\sim$-classes of $\mathbb{A}$, the mapping $c$ associates to each state $q$ its class $q/\sim=\{p:p\sim q\}$
and for two classes $p/\sim$ and $q/\sim$ there exists an edge from $p/\sim$ to $q/\sim$ labelled by $a\in\Sigma$
if and only if $p'a=q'$ for some $p'\sim p$, $q'\sim q$.
It is known that the component graph can be constructed from $\mathbb{A}$ in linear time.
Note that the mapping $c$ is redundant but it gives a possibility for determining whether $p\sim q$ holds in constant time on a
RAM machine, provided $Q=[n]$ for some $n>0$ and $c$ is stored as an array.

When $A$ and $B$ are sets, then $A^B$ denotes the set of all functions $f:B\to A$. When $f:B\to A$ and $C\subseteq B$,
then $f|_C:C\to A$ denotes the restriction of $f$ to $C$. When $A_1,\ldots,A_n$ are disjoint sets, $A$ is a set
and for each $i\in[n]$, $f_i:A_i\to A$ is a function, then the \emph{source tupling} of $f_1,\ldots,f_n$ is
the function $[f_1,\ldots,f_n]:\bigl(\mathop\bigcup\limits_{i\in[n]}A_i\bigr)\to A$ with $[f_1,\ldots,f_n](a)=f_i(a)$ for the
unique $i$ with $a\in A_i$.\marginpar{$[f_1,\ldots,f_n]$: source tupling}
Members of $Q^Q$ are called \emph{transformations} of $Q$, forming a semigroup with composition $(fg)(q)=g(f(q))$ as product.
When $\mathbb{A}=(Q,\Sigma,\delta,q_0,F)$ is an automaton, its \emph{transformation semigroup} $\mathcal{T}(\mathbb{A})$
consists of the set of transformations of $Q$ induced by nonempty words, i.e.
$\mathcal{T}(\mathbb{A})=\{u^\mathbb{A}:u\in\Sigma^+\}$ where $u^\mathbb{A}:Q\to Q$ is the transformation defined as $q\mapsto qu$.
A transformation $f:Q\to Q$ is called \emph{permutational} if there exists a set $D\subseteq Q$ with $|D|>1$ on which $f$ induces a permutation,
otherwise it's non-permutational.\marginpar{non-permutational transformation}
Observe that a non-permutational transformation $f$ is idempotent (i.e. $ff=f$) if and only if it is a constant function.
Alternatively, a transformation $f:Q\to Q$ is non-permutational for a finite $Q$ if and only if $f^{|Q|}$ is constant.
Another class of functions used in the paper is that of the \emph{elevating} functions: for the integers $0< k\leq n$, a function
$f:[k]\to[n]$ is elevating if $i<f(i)$ for each $i\in[k]$.\marginpar{elevating function}

\section{Patterns for subclasses of the star-free languages}
A language $L$ is 
\begin{itemize}
\item \emph{cofinite} if its complement is finite;
\item \emph{definite} if there exists a constant $k\geq 0$ such that for any $x\in\Sigma^*$, $y\in \Sigma^k$ we have $xy\in L\Leftrightarrow y\in L$;
\item \emph{reverse definite} if there exists a constant $k\geq 0$ such that for any $x\in\Sigma^k$, $y\in \Sigma^*$ we have $xy\in L\Leftrightarrow x\in L$;
\item \emph{generalized definite} if there exists a constant $k\geq 0$ such that for any $x_1,x_2\in\Sigma^k$ and $y\in\Sigma^*$ we have
  $x_1yx_2\in L\Leftrightarrow x_1x_2\in L$.
\end{itemize}
These are all subclasses of the star-free languages, i.e. can be built from the singletons with repeated use of the concatenation, finite union and complementation operations. It is known that the following decision problem is complete for $\mathbf{PSPACE}$: given a regular language $L$ with its minimal automaton, is $L$ star-free? In contrast, the question for these subclasses above are all tractable.

Minimal automata of the finite, cofinite, definite and reverse definite languages possess a characterization in terms of \emph{forbidden patterns}.
In our setting, a pattern is an edge-labelled, directed graph $P=(V,E,\ell)$, where $V$ is the set of vertices, $E\subseteq V^2$ is the set of edges,
and $\ell:E\to \mathcal{X}$ is a labelling function which assigns to each edge a variable.
An automaton $\mathbb{A}=(Q,\Sigma,\delta,q_0,F)$ \emph{admits a pattern $P=(V,E,\ell)$}\marginpar{admitting/avoiding a pattern} if there exists an \emph{injective} mapping $f:V\to Q$
and a map $h:\mathcal{X}\to \Sigma^+$ such that for each $(u,v)\in E$ labelled $x$ we have $f(u)\cdot h(x)=f(v)$.
Otherwise $\mathbb{A}$ \emph{avoids} $P$.
  
As an example, consider the pattern $P_f$ on Figure~\ref{fig-patterns}.

\begin{figure}[h!]
\centering
\subfloat[][Pattern $P_f$.]{
\begin{tikzpicture}[shorten >=1pt,node distance=2cm,>=stealth',thick]
\node[state] (1) {$p$};
\node[state] (2) [right of=1] {$q$};
\draw [->] (1) to[loop above] node[auto] {$x$} (1);
\draw [->] (2) to[loop above] node[auto] {$y$} (2);
\end{tikzpicture}
}
\hfil
\subfloat[][Pattern $P_d$.]{
\begin{tikzpicture}[shorten >=1pt,node distance=2cm,>=stealth',thick]
\node[state] (1) {$p$};
\node[state] (2) [right of=1] {$q$};
\draw [->] (1) to[loop above] node[auto] {$x$} (1);
\draw [->] (2) to[loop above] node[auto] {$x$} (2);
\end{tikzpicture}
}
\hfil
\subfloat[][Pattern $P_r$.]{
\begin{tikzpicture}[shorten >=1pt,node distance=2cm,>=stealth',thick]
\node[state] (1) {$p$};
\node[state] (2) [right of=1] {$q$};
\draw [->] (1) to[loop above] node[auto] {$x$} (1);
\draw [->] (1) to[bend right] node[auto] {$y$} (2);
\end{tikzpicture}
}
\caption{Patterns for (co)finite, definite and reverse definite languages.}
\label{fig-patterns}
\end{figure}
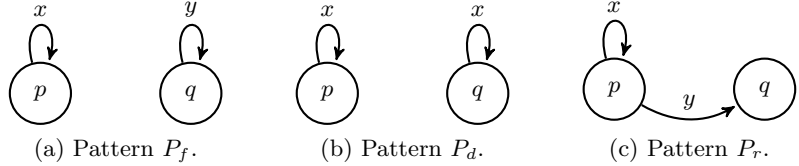

An automaton admits $P_f$ iff there exist \emph{different} states $p,q\in Q$ and (not necessarily different) words $x,y\in\Sigma^+$ such
that $px=p$ and $qy=q$. It is easy to see that an automaton $\mathbb{A}$ avoids $P_f$ iff it has a unique sink
which is a set consisting of a single state $p$, and all the other components are trivial;
if $p$ is a rejecting state, then $L(\mathbb{A})$ is finite, otherwise it is cofinite.
The condition is also necessary in the following sense: a language is finite or cofinite if and only if its minimal automaton avoids $P_f$. 

As other examples, consider the patterns $P_d$ and $P_r$ on Figure~\ref{fig-patterns}.

It is easy to see that if $\mathbb{A}=(Q,\Sigma,\delta,q_0,F)$ is the minimal automaton of a reverse definite language, then
it avoids $P_r$: if there are states $p\neq q\in Q$ and words $x,y\in\Sigma^+$ with $px=p$ and $py=q$,
then $L=L(\mathbb{A})$ is not reverse definite. Indeed, suppose $L$ is a $k$-reverse definite language and let $u$ be a word with
$q_0u=p$. Since $p\neq q$ and $\mathbb{A}$ is minimal, there is a word $w$ distinguishing $p$ and $q$. Thus,
$ux^kw$ and $ux^kyw$ are two words with the same prefix of length $k$, and exactly one of them is in $L$, a contradiction.

Also, if $L=L(\mathbb{A})$ is a $k$-definite language with $\mathbb{A}$ being its minimal automaton, then $\mathbb{A}$
avoids $P_d$: if there are states $p\neq q\in Q$ and a word $x$ with $px=p$, $qx=q$, then let $u,v,w\in\Sigma^*$ be words
such that $q_0u=p$, $q_0v=q$ and $w$ separates $p$ and $q$. Then $ux^kw$ and $vx^kw$ have the same suffix of length $k$,
with exactly one of them being a member of $L$, a contradiction.

It can be seen (see e.g. \cite{brzozo}) that avoiding these patterns are also sufficient: a regular language is definite
(reverse definite, resp.) if and only if its minimal automaton avoids $P_d$ ($P_r$, resp.).
Note that avoiding $P_d$ is equivalent to state that each nonempty word induces a transformation with at most one fixed point,
which is further equivalent to state that each nonempty word induces a non-permutational transformation. See~\cite{brzozo}\footnote{Since -- up to our knowledge -- ~\cite{brzozo} has not been published yet in a peer-reviewed journal or conference proceedings, we include a proof of this fact. Nevertheless, we do not claim this result to be ours, by any means.}.)

Consequently, all the following questions are in the complexity class $\mathbf{NL}$: given a language $L$ by its
minimal automaton, is $L$ (co)finite / definite / reverse definite?

\section{Results}

In this section we give a new characterization of the minimal automata of generalized definite languages,
leading to an $\mathbf{NL}$-completeness result of the corresponding decision problem,
as well as a low-degree polynomial deterministic algorithm,
and show that the syntactic complexity of generalized definite languages is the same as that of the definite languages.
We also give an upper bound $n!$ for the syntactic complexity of (generalized) definite languages.

\subsection{Forbidden pattern characterization}

We need the following well-known lemma:
\begin{lemma}
\label{lemma-power}
For any nonempty finite set $C$ there exists a constant $m=m(|C|)$ depending only on the size of $C$
such that in any product $f=f_1f_2\ldots f_m$ with $f_i\in C^C$ for each $i\in[m]$, an idempotent factor appears,
i.e. $f_j\ldots f_k$ is an idempotent transformation of $C$ for some $1\leq j\leq k\leq m$.
\end{lemma}
Note to the reviewers: we were unable to locate the first appearance with proof of Lemma~\ref{lemma-power}, thus we decided
to include its proof in the Appendix.

We are ready to show that a regular language is generalized definite if and only if its minimal automaton avoids the pattern $P_g$, depicted on Figure~\ref{fig-gen}.
\begin{figure}[h!]
\centering
\scalebox{0.8}{
\begin{tikzpicture}[shorten >=1pt,node distance=2cm,>=stealth',thick]
\node[state] (1) {$p$};
\node[state] (2) [right of=1] {$q$};
\draw [->] (1) to[loop above] node[auto] {$x$} (1);
\draw [->] (2) to[loop above] node[auto] {$x$} (2);
\draw [->] (1) to[bend right] node[auto] {$y$} (2);
\end{tikzpicture}
}
\caption{Forbidden pattern $P_g$ for the generalized definite languages.}
\label{fig-gen}
\end{figure}
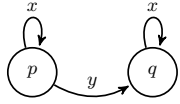

\begin{theorem}
\label{thm-pattern}
The following are equivalent for a reduced automaton $\mathbb{A}$:
\begin{enumerate}
\item[i)] $\mathbb{A}$ avoids $P_g$.
\item[ii)] Each nontrivial component of $\mathbb{A}$ is a sink, and for each nonempty word $u$ and sink $C$ of $\mathbb{A}$,
  the transformation $u|_C:C\to C$ is non-permutational.
\item[iii)] $\mathbb{A}$ recognizes a generalized definite language.
\end{enumerate}
\end{theorem}
\begin{proof}
Let $\mathbb{A}=(Q,\Sigma,\delta,q_0,F)$ be a reduced automaton.

{\bf i)$\to$ii).} Suppose $\mathbb{A}$ avoids $P_g$.
Suppose that $u|_C$ is permutational for some sink $C$ and word $u\in\Sigma^+$.
Then there exists a set $D\subseteq C$ with $|D|>1$ such that $u$ induces a permutation
on $D$. Then, $x=u^{|D|!}$ is the identity on $D$. Choosing arbitrary distinct states $p,q\in D$
and a word $y$ with $py=q$ (such $y$ exists since $p$ and $q$ are in the same component of $\mathbb{A}$),
we get that $\mathbb{A}$ admits $P_g$ by the $(p,q,x,y)$ defined above, a contradiction.
Hence, $u|_C$ is non-permutational for each sink $C$ and word $u\in\Sigma^+$.

Now assume there exists a nontrivial component $C$ which is not a sink.
Then, $pu=p$ for some $p\in C$ and word $u\in\Sigma^+$. Since $C$ is not a sink, there exists
a sink $C'\neq C$ reachable from $p$ (i.e. all of its members are reachable from $p$).
Since $u$ induces a non-permutational transformation on $C'$, $x=u^{|C'|}$ induces a constant function
on $C'$. Let $q$ be the unique state in the image of $x|_{C'}$. Since $C'$ is reachable from $p$,
there exists some nonempty word $y$ such that $py=q$. Hence, $px=p$, $qx=q$, $py=q$ and $\mathbb{A}$
admits $P_g$, a contradiction.

{\bf ii)$\to$iii).} Suppose the condition of ii) holds.
We show that $L(\mathbb{A})$ is generalized definite. Let $n=m(|Q|)$ be the value defined in Lemma~\ref{lemma-power}.
Let $x=x_1yx_2$ with $x_1,x_2\in\Sigma^n$, $y\in\Sigma^*$.
It suffices to show that $q_0x_1yx_2=q_0x_1x_2$.
Since $|x_1|\geq|Q|$, some state $p$ is visited at least twice on the path determined by $x_1$.
Hence $p$ belongs to a nontrivial component $C$ of $\mathbb{A}$, which has to be a sink by the assumption of ii).
Thus, $q_0x_1\in C$ and $q_0x_1y\in C$ as well.
By Lemma~\ref{lemma-power}, $x_2$ can be written as 
$x_2=x_{2,1}x_{2,2}x_{2,3}$ with $x_{2,2}$ inducing an idempotent function on $C$.
Since the function induced by $x_{2,2}$ is also non-permutational on $C$, it is a constant function on $C$,
hence $x_2$ induces a constant function as well. Thus $px_2=pyx_2$ and hence $q_0x_1yx_2=q_0x_1x_2$.

{\bf iii)$\to$i).} Suppose $L(\mathbb{A})$ is $k$-generalized definite for some $k>0$ and that $\mathbb{A}$ admits $P_g$,
i.e. $px=p$, $qx=q$ and $py=q$ for some distinct states $p,q$ and nonempty words $x,y$.
Since $\mathbb{A}$ is reduced, $p=q_0u$ for some $u\in\Sigma^*$, and there exists a word $w$ distinguishing $p$ and $q$.
Considering the words $ux^kx^kw$ and $ux^kyx^kw$ we get that they have the same prefix and suffix of length $k$,
but exactly one of them is a member of $L(\mathbb{A})$, a contradiction.
\end{proof}

\subsection{Complexity issues}

Using the characterization given in Theorem~\ref{thm-pattern}, we study the complexity of the following decision problem $\textsc{GenDef}$: given a finite automaton $\mathbb{A}$, is $L(\mathbb{A})$ a generalized definite language?

\begin{theorem}
Problem $\textsc{GenDef}$ is $\mathbf{NL}$-complete.
\end{theorem}
\begin{proof}
First we show that $\textsc{GenDef}$ belongs to $\mathbf{NL}$.
By~\cite{cho-huynh}, minimizing a DFA can be done in nondeterministic logspace.
Thus we can assume that the input is already minimized, since the class of
(nondeterministic) logspace computable functions is closed under composition.

Consider the following algorithm:
\begin{enumerate}
\item Guess two different states $p$ and $q$.
\item Let $s:=p$.
\item Guess a letter $a\in \Sigma$. Let $s:=sa$.
\item If $s=q$, proceed to Step 5. Otherwise go back to Step 3.
\item Let $p':=p$ and $q':=q$.
\item Guess a letter $a\in\Sigma$. Let $p':=p'a$ and $q'=q'a$.
\item If $p=p'$ and $q=q'$, accept the input. Otherwise go back to Step 6.
\end{enumerate}
The above algorithm checks whether $\mathbb{A}$ admits $P_g$: first it guesses $p\neq q$,
then in Steps 2--4 it checks whether $q$ is accessible from $p$,
and if so, then in Steps 5--7 it checks whether there exists a word $x\in\Sigma^+$ with $px=p$ and $qx=q$.
Thus it decides\footnote{Note that in this form, the algorithm can enter an infinite loop which fits
into the definition of nondeterministic log\emph{space}.
Introducing a counter and allowing at most $n$ steps in the first cycle and at most $n^2$ in the second
we get a nondeterministic algorithm using logspace and polytime, as usual.} the complement of $\textsc{GenDef}$, in nondeterministic logspace; since $\mathbf{NL}=\mathrm{co}\mathbf{NL}$,
we get that $\textsc{GenDef}\in\mathbf{NL}$ as well.

For $\mathbf{NL}$-completeness we recall from~\cite{jones-lien-laaser} that the reachability problem for DAGs ($\textsc{DAG-Reach}$)
is complete for $\mathbf{NL}$:
given a directed acyclic graph $G=(V,E)$ on $V=[n]$ with $(i,j)\in E$ only if $i<j$,
is $n$ accessible from $1$?
We give a logspace reduction from $\textsc{DAG-Reach}$ to $\textsc{GenDef}$ as follows.
Let $G=([n],E)$ be an instance of $\textsc{DAG-Reach}$.
For a vertex $i\in[n]$, let $N(i)=\{j:(i,j)\in E\}$ stand for the set of its neighbours and
let $d(i)=|N(i)|<n$ denote the outdegree of $i$. When $j\in[d(i)]$, then the $j$th neighbour of $i$, denoted $n(i,j)$
is simply the $j$th element of $N(i)$ (with respect to the usual ordering of integers of course).
Note that for any $i\in [n]$ and $j\in[d(i)]$ both $d(i)$ and the $n(i,j)$ (if exists) can be computed in logspace.

We define the automaton $\mathbb{A}=([n+1],[n],\delta,1,\{n+1\})$
where 
\[\delta(i,j)=\left\{\begin{array}{ll}
n+1&\hbox{if }(i=n+1)\hbox{ or }(j=n)\hbox{ or }(i<n\hbox{ and }d(i)<j);\\
1&\hbox{if }i=n\hbox{ and }j<n;\\
n(i,j)&\hbox{otherwise.}\\
\end{array}\right.\]
Note that $\mathbb{A}$ is indeed an automaton, i.e. $\delta(i,j)$ is well-defined for each $i,j$.

We claim that $\mathbb{A}$ admits $P_g$ if and only if $n$ is reachable from $1$ in $G$.
Observe that the underlying graph of $\mathbb{A}$ is $G$, with a new edge $(n,1)$ and with a new vertex $n+1$, which is a neighbour of each vertex.
Hence, $\{n+1\}$ is a sink of $\mathbb{A}$ which is reachable from all other states.
Thus $\mathbb{A}$ admits $P_g$ if and only if there exists a nontrivial component of $\mathbb{A}$ which is different from $\{n+1\}$.
Since in $G$ there are no cycles, such component exists if and only if the addition of the edge $(n,1)$ introduces a cycle,
which happens exactly in the case when $n$ is reachable from $1$. Note that it is exactly the case when $1x=1$ for some word $x\in\Sigma^+$.

What remains is to show that the \emph{reduced} form $\mathbb{B}$ of $\mathbb{A}$ admits $P_g$ if and only if $\mathbb{A}$ does.
First, both $1$ and $n+1$ are in the connected part $\mathbb{A}'$ of $\mathbb{A}$, and are distinguishable by the empty word
(since $n+1$ is final and $1$ is not).
Thus, if $\mathbb{A}$ admits $P_g$ with $1x=1$ and $(n+1)x=n+1$ for some $x\in\Sigma^+$, then $\mathbb{B}$ admits $P_g$ with
$h(1)x=h(1)$ and $h(n+1)x=h(n+1)$ (with $h$ being the homomorphism from the connected part of $\mathbb{A}$ onto its reduced form).
For the other direction, assume $h(p)x_0=h(p)$ for some state $p\neq n+1$
(note that since $n+1$ is the only final state, $p\neq n+1$ if and only if $h(p)\neq h(n+1)$).
Let us define the sequence $p_0,p_1,\ldots$ of states of $\mathbb{A}$ as $p_0=p$, $p_{t+1}=p_tx_0$.
Then, for each $i\geq 0$, $h(p_i)=h(p)$, thus $p_i\in[n]$. Thus, there exist indices $0\leq i<j$ with $p_i=p_j$, yielding
$p_ix_0^{j-i}=p_i$, thus $\mathbb{A}$ admits $P_g$ with $p=p_i$, $q=n+1$, $x=x_0^{j-i}$ and $y=n$.

Hence, the above construction is indeed a logspace reduction
from $\textsc{DAG-Reach}$ to the complement of $\textsc{GenDef}$, showing $\mathbf{NL}$-hardness of the latter;
applying $\mathbf{NL}=\mathrm{co}\mathbf{NL}$ again, we get $\mathbf{NL}$-hardness of $\textsc{GenDef}$ itself.
\end{proof}

It is worth observing that the same construction also shows $\mathbf{NL}$-hardness (thus completeness) of the
problem whether the input automaton accepts a definite language.

Thus, the complexity of the problem is characterized from the theoretic point of view.
However, nondeterministic algorithms are not that useful in practice. Since $\mathbf{NL}\subseteq\mathbf{P}$, the
problem is solvable in polynomial time -- now we give an efficient (quadratic) deterministic decision algorithm:

\begin{enumerate}
\item Compute $\mathbb{A}'=(Q,\Sigma,\delta,q_0,F)$, the reduced form of the input automaton $\mathbb{A}$.
\item Compute $\Gamma(\mathbb{A}')$, the component graph of $\mathbb{A}'$.
\item If there exists a nontrivial, non-sink component, reject the input.
\item Compute $\mathbb{B}=\mathbb{A}'\times\mathbb{A}'$ and $\Gamma(\mathbb{B})$.
\item Check whether there exist a state $(p,q)$ of $\mathbb{B}$ in a nontrivial component (of $\mathbb{B}$)
  for some $p\neq q$ with $p$ being in the same sink as $q$ in $\mathbb{A}$. If so, reject the input; otherwise accept it.
\end{enumerate}

The correctness of the algorithm is straightforward by Theorem~\ref{thm-pattern}: after minimization
(which takes $\mathcal{O}(n\log n)$ time) one computes the component graph of the reduced automaton
(taking linear time) and checks whether there exists a nontrivial component which is not a sink
(taking linear time again, since we already have the component graph). If so, then the answer is $\texttt{NO}$.
Otherwise one has to check whether there is a (sink) component $C$ and a word $x\in\Sigma^+$ such that
$f_x|_C$ has at least two different fixed points. Now it is equivalent to ask
whether there is a state $(p,q)$ in $\mathbb{A}'\times\mathbb{A}'$ with $p$ and $q$ being in the same
component and a word $x\in\Sigma^+$ with $(p,q)x=(p,q)$. This is further equivalent to ask whether
there is a $(p,q)$ with $p,q$ being in the same sink such that $(p,q)$ is in a nontrivial component of $\mathbb{B}$.
Computing $\mathbb{B}$ and its components takes $\mathcal{O}(n^2)$ time, and (since we still have the component graph of $\mathbb{A}$)
checking this condition takes constant time for each state $(p,q)$ of $\mathbb{B}$, the algorithm consumes a total of $\mathcal{O}(n^2)$
time.

Hence we have an upper bound concluding this subsection:
\begin{theorem}
Problem $\textsc{GenDef}$ can be solved in $\mathcal{O}(n^2)$ deterministic time in the RAM model of computation.
\end{theorem}

\subsection{Syntactic complexity}
The \emph{syntactic complexity} of a language is the size of its syntactic semigroup, the latter being isomorphic to the transformation semigroup
$\mathcal{T}(\mathbb{A})$ of the minimal automaton $\mathbb{A}$ of the language (equipped with function composition
as product).
The \emph{syntactic complexity} of a \emph{class} $\mathcal{C}$ of regular languages is a function $n\mapsto f(n)$ where $f(n)$ is the maximal syntactic complexity a member of $\mathcal{C}$ can have whose minimal automaton has at most $n$ states.

In~\cite{brzozo} it has been shown that the class of definite languages has syntactic complexity $\geq \lfloor e\cdot(n-1)!\rfloor$,
thus the same lower bound also applies for the larger class of generalized definite languages.

\begin{theorem}
The syntactic complexity of the definite and that of the generalized definite languages coincide.
\end{theorem}
\begin{proof}
It suffices to construct for an arbitrary reduced automaton $\mathbb{A}=(Q,\Sigma,\delta,q_0,F)$ recognizing a generalized definite language
a reduced automaton $\mathbb{B}=(Q,\Delta,\delta',q_0,F')$ for some $\Delta$
recognizing a definite language such that $|\mathcal{T}(\mathbb{A})|\leq |\mathcal{T}(\mathbb{B})|$.

By Theorem~\ref{thm-pattern}, if $L(\mathbb{A})$ is generalized definite and $\mathbb{A}$ is reduced, then $Q$ can be partitioned as a disjoint union
$Q=Q_0\uplus Q_1\uplus\ldots\uplus Q_c$ for some $c>0$ such that each $Q_i$ with $i\in[c]$ is a sink of $\mathbb{A}$ and $Q_0$ is the
(possibly empty) set of those states that belong to a trivial component. Without loss of generality we can assume that
$Q=[n]$ and $Q_0=[k]$ for some $n$ and $k$, and that for each $i\in[k]$ and $a\in \Sigma$, $i<ia$.
The latter condition is due to the fact that reachability restricted to the set $Q_0$ of states in trivial components is a
partial ordering of $Q_0$ which can be extended to a linear ordering.
Clearly, if $Q_0$ is nonempty, then by connectedness $q_0=1$ has to hold; otherwise $c=1$ and we again may assume $q_0=1$.
Also, $Q_i\Sigma\subseteq Q_i$ for each $i\in[c]$, and let $|Q_1|\leq |Q_2|\leq\ldots\leq|Q_c|$.

Then, each transformation $f:Q\to Q$ can be uniquely written as the source tupling $[f_0,\ldots,f_c]$ of some functions $f_i:Q_i\to Q$
with $f_i:Q_i\to Q_i$ for $0<i\leq c$.
For any $[f_0,\ldots,f_c]\in\mathcal{T}=\mathcal{T}(\mathbb{A})$ the following hold: $f_0(i)>i$ for each $i\in[k]$, and $f_j$ is non-permutational
on $Q_j$ for each $j\in[c]$.
For $k=0,\ldots,c$, let $\mathcal{T}_k$ stand for the set $\{f_k:f\in\mathcal{T}\}$ (i.e. the set of functions $f|_{Q_k}$ with
$f\in\mathcal{T}$). Then, $|\mathcal{T}|\leq \mathop\prod\limits_{0\leq k\leq c}|\mathcal{T}_k|$.

If $|Q_c|=1$, then all the sinks of $\mathbb{A}$ are singleton sets.
Thus there are at most two sinks, since if $C$ and $D$ are singleton sinks whose
members do not differ in their finality, then their members are not distinguishable, thus $C=D$ since $\mathbb{A}$ is reduced.
Such automata recognize reverse definite languages,
having a syntactic semigroup of size at most $(n-1)!$ by \cite{brzozo}, 
thus in that case $\mathbb{B}$ can be chosen to an arbitrary definite automaton having $n$ state and a syntactic semigroup of size
at least $\lfloor e(n-1)!\rfloor$ (by the construction in \cite{brzozo}, such an automaton exists).
Thus we may assume that $|Q_c|>1$. (Note that in that case $Q_c$ contains at least one final and at least one non-final state.)

Let us define the sets $\mathcal{T}'_k$ of functions $Q_i\to Q$ as $\mathcal{T}'_0$ is the set of all elevating functions from $[k]$ to $[n]$,
$\mathcal{T}'_c=\mathcal{T}_c$ and for each $0<k<c$, $\mathcal{T}'_k=Q_c^{Q_k}$. Since $\mathcal{T}_k\subseteq Q_k^{Q_k}$ and $|Q_k|\leq |Q_c|$
for each $k\in[c]$, we have $|\mathcal{T}_k|\leq|\mathcal{T}'_k|$ for each $0\leq k\leq c$. Thus defining $\mathcal{T}'=\{[f_0,\ldots,f_c]:f_i\in\mathcal{T}'_i\}$ it holds that $|\mathcal{T}|\leq|\mathcal{T}'|$.

We define $\mathbb{B}$ as $(Q,\mathcal{T}',\delta',q_0,F)$ with $\delta'(q,f)=f(q)$ for each $f\in\mathcal{T}'$. We show that
$\mathbb{B}$ is a reduced automaton avoiding $P_d$, concluding the proof.

First, observe that $\mathbb{B}$ has exactly one sink, $Q_c$, and all the other states belong to trivial components
(since by each transition, each member of $Q_0$ gets elevated, and each member of $Q_i$ with $0<i<c$ is taken into $Q_c$).
Hence if $\mathbb{B}$ admits $P_d$, then $pt=p$ and $qt=q$ for some distinct pair $p,q\in Q_c$ of states and $t=[t_0',\ldots,t_c']\in\mathcal{T}'$.
This is further equivalent to $pt'_c=p$ and $qt'_c=q$ for some $p\neq q$ in $Q_c$ and $t'_c\in\mathcal{T}'_c$.
By definition of $\mathcal{T}'_c=\mathcal{T}_c$, there exists a transformation of the form $t=[t_0,\ldots,t_{c-1},t'_c]\in\mathcal{T}$
induced by some word $x$, thus $px=p$ and $qx=q$ both hold in $\mathbb{A}$,
and since $p,q$ are in the same sink, there also exists a word $y$ with $py=q$. Hence $\mathbb{A}$ admits $P_g$, a contradiction.

Second, $\mathbb{B}$ is connected. To see this, observe that each state $p\neq 1$ is reachable from $1$ by any transformation of the form
$t=[f_p,t_1,\ldots,t_c]$ where $f_p:[k]\to[n]$ is the elevating function with $1f_p=p$ and $if_p=n$ for each $i>1$.
Of course $1$ is also trivially reachable from itself, thus $\mathbb{B}$ is connected.

Also, whenever $p\neq q$ are different states of $\mathbb{B}$, then they are distinguishable by some word.
To see this, we first show this for $p,q\in Q_c$. Indeed, since $\mathbb{A}$ is reduced, some transformation $t=[t_0,\ldots,t_c]\in\mathcal{T}$
separates $p$ and $q$ (exactly one of $pt=pt_c$ and $qt=qt_c$ belong to $F$). Since $\mathcal{T}_c=\mathcal{T}'_c$, we get that
$p$ and $q$ are also distinguishable by in $\mathbb{B}$ by any transformation of the form $t'=[t_0',\ldots,t_{c-1}',t_c]\in\mathcal{T}'$.
Now suppose neither $p$ nor $q$ belong to $Q_c$. Then, since $\{[t_0',\ldots,t_{c-1}']:t_i'\in\mathcal{T}_i'\}=Q_c^{Q\backslash Q_c}$,
and $|Q_c|>1$, there exists some $t=[t_0',\ldots,t_{c-1}']$ with $pt\neq qt$, thus any transformation of the form
$[t_0',\ldots,t_{c-1}',t_c]\in\mathcal{T}'$ maps $p$ and $q$ to distinct elements of $Q_c$, which are already known to be distinguishable,
thus so are $p$ and $q$. Finally, if $p\in Q_c$ and $q\notin Q_c$, then let $t_c\in\mathcal{T}_c$ be arbitrary and
$t'=[t_0',\ldots,t_{c-1}]\in Q_c^{Q\backslash Q_c}$ with $qt'\neq pt_c$. Then $[t',t_c]$ again maps $p$ and $q$ to distinct states of $Q_c$.

Thus $\mathbb{B}$ is reduced, concluding the proof: $\mathbb{B}$ is a reduced automaton recognizing a definite language
and having a syntactic semigroup $\mathcal{T}'$ with $|\mathcal{T}'|\geq|\mathcal{T}|$.
\end{proof}

\subsection{Upper bound for syntactic complexity}
By \cite{brzozo} we know a lower bound $\lfloor e(n-1)!\rfloor$ for the syntactic complexity of the definite languages
(thus, of the generalized definite ones as well). In this subsection we give an upper bound $n!$, showing that the bound of \cite{brzozo}
is asymptotically optimal up to a logarithmic factor (since $n=\mathcal{O}(\log n!)$).

Let $\mathbb{A}=(Q,\Sigma,\delta,q_0,F)$ be a reduced automaton recognizing a definite language $L$ and let $\mathcal{T}\subseteq Q^Q$
be its syntactic semigroup. Then, each member $t$ of $\mathcal{T}$ is non-permutational and has a unique fixed point $\mathrm{fix}(t)$.
For each $p\in Q$, let $\mathcal{T}_p$ stand for the subset $\{t\in\mathcal{T}:\mathrm{fix}(t)=p\}$ of $\mathcal{T}$: then, $\mathcal{T}$
is the disjoint union of the sets $\mathcal{T}_p$. Observe that $\mathcal{T}_p$ is a semigroup for each $p$, since whenever
$\mathrm{fix}(t)=\mathrm{fix}(t')=p$, then $ptt'=p$, thus $p$ is a fixed point of $tt'$ (and by assumption, the superset
$\mathcal{T}$ of $\mathcal{T}_p$ is a semigroup consisting only non-permutational transformations). Thus $tt'\in\mathcal{T}_p$ as well.

\begin{lemma}
For each $p\in Q$, $|\mathcal{T}_p|\leq (n-1)!$.
\end{lemma}
\begin{proof}
Let $G_p=(Q,E,\ell)$ be the edge-labelled graph on the set $Q$ of vertices in which $(q_1,q_2)$ is an edge labelled by $t\in\mathcal{T}_p$
if and only if $q_1t=q_2$ and $q_1\neq p$. Then $G_p$ is acyclic.

Indeed, suppose $q_1\mathop{\to}\limits^{t_1}q_2\mathop{\to}\limits^{t_2}\ldots\mathop{\to}\limits^{t_k}q_{k+1}=q_1$.
Then $q_1t_1t_2\ldots t_k=q_1$, thus $q_1$ is a fixed point of $t=t_1\ldots t_k\in \mathcal{T}_p$. Since in $G_p$ the vertex $p$
has outdegree $0$, $q_0\neq p$, hence $t$ has at least two distinct fixed points, a contradiction. Hence $G_p$ is acyclic.
Thus, there exists an ordering $\prec$ on $Q$ such that whenever $q_1t=q_2$ for some $q_1,q_2\in Q$, $q_1\neq p$ and $t\in\mathcal{T}_p$,
then $q_1\prec q_2$. Note also that $p$ is the maximal element of $\prec$.
Thus $\mathcal{T}_p$ consists of transformations $t:Q\to Q$ with $pt=p$, and $q\prec qt$ for each $q\in Q-\{p\}$. There are
$(n-1)!$ such transformations (the least element can be mapped to the other $n-1$ elements, the next to $n-2$ and so on),
concluding the lemma.
\end{proof}
\begin{corollary}
The syntactic complexity of definite languages is at most $n!$.
\end{corollary}
\begin{proof}
For an arbitrary automaton $\mathbb{A}$ over $n$ states recognizing a definite language, $\mathcal{T}(\mathbb{A})=\mathop\bigcup_{p\in Q}\mathcal{T}_p$,
hence its size is at most $n\cdot (n-1)! = n!$.
\end{proof}

\section{Conclusion, further directions}
The forbidden pattern characterization of generalized definite languages we gave is not surprising, based on the identities of the pseudovariety
of (syntactic) semigroups corresponding to this variety of languages.
Still, using this characterization one can derive efficient algorithms for checking whether a given automaton recognizes such a language.
Though we could not compute an exact function for the syntactic complexity, we still managed to show that
these languages are not ``more complex'' than definite languages under this metric. Also, we gave a new upper bound for that.

The exact syntactic complexity of definite languages is still open, as well as for other language classes higher in the dot-depth hierarchy --
e.g. the locally (threshold) testable and the star-free languages.

\section*{Appendix}
In the Appendix we give a proof of Lemma~\ref{lemma-power} and that a regular language $L$ is definite if and only if its minimal automaton avoids $P_d$.

We will make use of the following variant of the multicolor Ramsey theorem, stated here only for monochromatic triangles.
\begin{theorem}
\label{thm-ramsey}
For any number $c>0$ of colors there exists an integer $R(c)$ such that whenever $G$ is an edge-colored complete graph
on at least $R(c)$ vertices that has at most $c$ colors, then $G$ contains a monochromatic triangle.
\end{theorem}
The theorem holds for monochromatic arbitrary-sized induced subgraphs as well but we need only the guaranteed appearance of
triangles to show that in a finite semigroup, a long enough product always has an idempotent factor.

\begin{proof}[of Lemma~\ref{lemma-power}]
Let $m=R(|C^C|)$ and let us define the following complete graph on $[m]$ with its edges colored by elements of $C^C$: let the color of the edge
$(i,j)$, $i<j$, be the element $f_{i,j}=f_if_{i+1}\ldots f_{j-1}\in C^C$.
Applying Theorem~\ref{thm-ramsey} we get that there exists integers $1\leq i<j<k\leq m$ with $(i,j)$, $(j,k)$ and $(i,k)$
having the same color, i.e. $f_{i,j}=f_{j,k}=f_{i,k}$, the last being the product of $f_{i,j}$ and $f_{j,k}$.
Hence, $f_{i,j}$ is an idempotent transformation of $C$.
\end{proof}

Now for the forbidden pattern characterization of definite languages:
\begin{theorem}
The following are equivalent for a reduced automaton $\mathbb{A}=(Q,\Sigma,\delta,q_0,F)$:
\begin{itemize}
\item[i)] $L(\mathbb{A})$ is definite.
\item[ii)] $\mathbb{A}$ avoids $P_d$.
\item[iii)] For each $u\in\Sigma^+$, $u^\mathbb{A}$ is non-permutational.
\item[iv)] $\mathbb{A}$ has a unique sink $C$,
  all its other components are trivial and for each $u\in\Sigma^+$, $u^\mathbb{A}|_C$ is non-permutational.
\end{itemize}
\end{theorem}
\begin{proof}
{\bf i)$\to$ii).} Assume $L=L(\mathbb{A})$ is $k$-definite for some $k>0$, and
   $\mathbb{A}$ admits $P_d$ with $px=p$ and $qx=q$ for distinct states $p,q$ and word $x\in\Sigma^+$.
   Since $\mathbb{A}$ is reduced, $q_0z_p=p$ and $q_0z_q=q$ for some words $z_p,z_q$ and $p,q$ are distinguishable by some word $w$.
   Then, exactly one of the words $z_px^kw$ and $z_qx^kw$ belongs to $L$ but they share a common suffix of length $k$, a contradiction.

{\bf ii)$\to$iii).} Assume $u^\mathbb{A}$ is permutational for some $u\in\Sigma^+$. Let $D\subseteq Q$, $|D|>1$ be a set on which
  $u$ induces a permutation. Then $u^{|D|!}$ induces the identity on $D$, thus $\mathbb{A}$ admits $P_d$ with arbitrary $p,q\in D$
  and $x=u^{|D|!}$.
  
{\bf iii)$\to$iv).}  Obviously $\mathbb{A}$ has a sink $C$.
  If $u^\mathbb{A}$ is non-permutational for each $u\in\Sigma^+$, then $u^\mathbb{A}|_C$ is also non-permutational for each sink $C$.
  Hence, $u^{|C|}$ induces a constant function on $C$.
  Assume that there exists another nontrivial component $D\neq C$ of $\mathbb{A}$.
  Then $px_0=p$ for some $p\in D$ and $x_0\in\Sigma^+$. Thus, $x_0^{|C|}$ induces a permutational transformation on $Q$ (with fixed
  points $p\in D$ and the unique element of $Cx_0^{|C|}$), a contradiction.
  
{\bf iv)$\to$i).} Analogously to the direction ii)$\to$iii) of the proof of Theorem~\ref{thm-pattern}.
  Suppose the condition of iv) holds.
  Let $n=\max\{m(|Q|),|Q|\}$ be the value defined in Lemma~\ref{lemma-power}.
  Let $x=yx_2$ with $x_2\in\Sigma^n$, $y\in\Sigma^*$.
  It suffices to show that $q_0yx_2=q_0x_2$.
  Since $n\geq |Q|$, both $q_0yx_2$ and $q_0x_2$ belong to the unique sink $C$ of $\mathbb{A}$.
  By Lemma~\ref{lemma-power}, $x_2$ can be written as 
  $x_2=x_{2,1}x_{2,2}x_{2,3}$ with $x_{2,2}$ inducing an idempotent function on $C$.
  Since the function induced by $x_{2,2}$ is also non-permutational on $C$, it is a constant function on $C$,
  hence $x_2$ induces a constant function as well. Thus $q_0yx_2=q_0x_2$ and $L(\mathbb{A})$ is $n$-definite.
\end{proof}

\end{document}